\documentclass[times]{IEEEtranTCOM}

\usepackage{moreverb}
\usepackage{graphicx, subfigure}
\usepackage{epstopdf}
\usepackage{amsmath}
\usepackage{pseudocode}
\newtheorem{lemma}{Lemma}
\usepackage[colorlinks,bookmarksopen,bookmarksnumbered,citecolor=red,urlcolor=red]{hyperref}
\usepackage[numbers, square]{natbib}

\newcommand\BibTeX{{\rmfamily B\kern-.05em \textsc{i\kern-.025em b}\kern-.08em
T\kern-.1667em\lower.7ex\hbox{E}\kern-.125emX}}

\begin{document}

%\runningheads{Research Article (Review)}

%\articletype{RESEARCH ARTICLE}

\title{Memory Efficient Quasi-Cyclic Spatially Coupled  LDPC Codes}

\author{Vikram Arkalgud Chandrasetty, Sarah J. Johnson and Gottfried Lechner}

%\address{V.A. Chandrasetty and S.J. Johnson are with the University of Newcastle, Australia (e-mail: vikram.chandrasetty@newcastle.edu.au and sarah.johnson@newcastle.edu.au). G. Lechner is with the Institute for Telecommunications Research at the University of South Australia (email: gottfried.lechner@unisa.edu.au)}

%\corraddr{Journals Production Department, John Wiley \& Sons, Ltd,
%The Atrium, Southern Gate, Chichester, West Sussex, PO19~8SQ, UK.}

\maketitle
\begin{abstract}
In this paper we propose the construction of Spatially Coupled Low-Density Parity-Check (SC-LDPC) codes using a periodic time-variant Quasi-Cyclic (QC) algorithm. The QC based approach is optimized to obtain memory efficiency in storing the parity-check matrix in the decoders. A hardware model of the parity-check storage units has been designed for Xilinx FPGA to compare the logic and memory requirements for various approaches. It is shown that the proposed QC SC-LDPC code (with optimization) can be stored with reasonable logic resources and without the need of block memory in the FPGA. In addition, a significant improvement in the processing speed is also achieved.
\end{abstract}

%\keywords{class file; \LaTeXe; \emph{\journalabb}}
%Error correction codes, low-density paritycheck
%codes, spatially coupled codes, digital circuits, field programmable
%gate arrays.

%\footnotetext[2]{Please ensure that you use the most up to date
%class file, available from the ETT Home Page at}

\section{Introduction}
Spatially coupled Low-Density Parity-Check (SC-LDPC) codes, which can be thought of as a class of terminated LDPC convolutional codes \cite{Jimenez99}, have recently drawn significant interest in channel coding due to their excellent sum-product decoding thresholds \cite{Tanner2004,Lentmaier_ISIT05,Lentmaier_IT10,Kudekar_IT11}. Their excellent threshold performance is achieved at large code lengths (generally over 100K) which represents a significant challenge for implementation. The practical implementation of such large LDPC codes is a well known problem, particularly the storage of the parity-check matrix in the hardware \cite{Chandrasetty2012a} which can be achieved efficiently only for very structured parity-check matrices. The structure of the matrix also significantly affects the implementation complexity of the encoder and decoder. Quasi-Cyclic (QC) based LDPC matrices \cite{Tanner2001,Johnson_IrregLDPC,Fossorier2004} have proven advantages over unstructured (random) matrices in design complexity and encoding process \cite{Mahdi2011,Xinmiao2011}. They also enable collision-free parallel processing in the decoder \cite{Yongmei2008}.

Convolutional LDPC codes can be of two types: time-variant and time-invariant \cite{Jimenez99}. Recently, there have been studies on deriving time-variant and time-invariant LDPC convolutional codes by unwrapping QC LDPC block codes \cite{Pusane2011}. It has also been noted that the time-invariant codes are less complex for implementation but the decoding performance is poor compared to time-variant codes \cite{Chiu2013}.

In this paper, we present a special case of periodic time-variant SC-LDPC codes using a QC construction technique. The inherent advantages of QC based codes and the diagonal structure of the SC-LDPC matrix are exploited to reduce the complexity of the decoder by reusing the circulants in the matrix. When using sum-product decoding, a critical factor in the decoder performance is the girth of the code in the Tanner graph. Consequently, in this paper we will study the impact of the period of time-varying SC-LDPC codes on the girth of their Tanner graphs.
A comprehensive analysis is carried out to evaluate the performance in terms of bit error rate (BER) and hardware implementation complexity of these memory-optimized QC SC-LDPC codes. The FPGA resource requirements and speed of operation are compared by implementing a hardware model of SC-LDPC codes using QC \cite{Hagiwara2011} and progressive edge growth (PEG) \cite{Zhengang05} techniques. Memory efficiency and speed improvements achievable by using the proposed optimized QC SC-LDPC codes are also presented.

%The complexity versus performance with respect to the periodicity \cite{Jimenez99} of these codes is not addressed yet.

\begin{figure}
    \begin{center}
  \subfigure[A regular $(3,6)$ LDPC protograph with $n_b=2$, $n_c=1$.]
  {
   \label{fig:sc_ldpc_protograph}
   \includegraphics[width=0.19\textwidth]{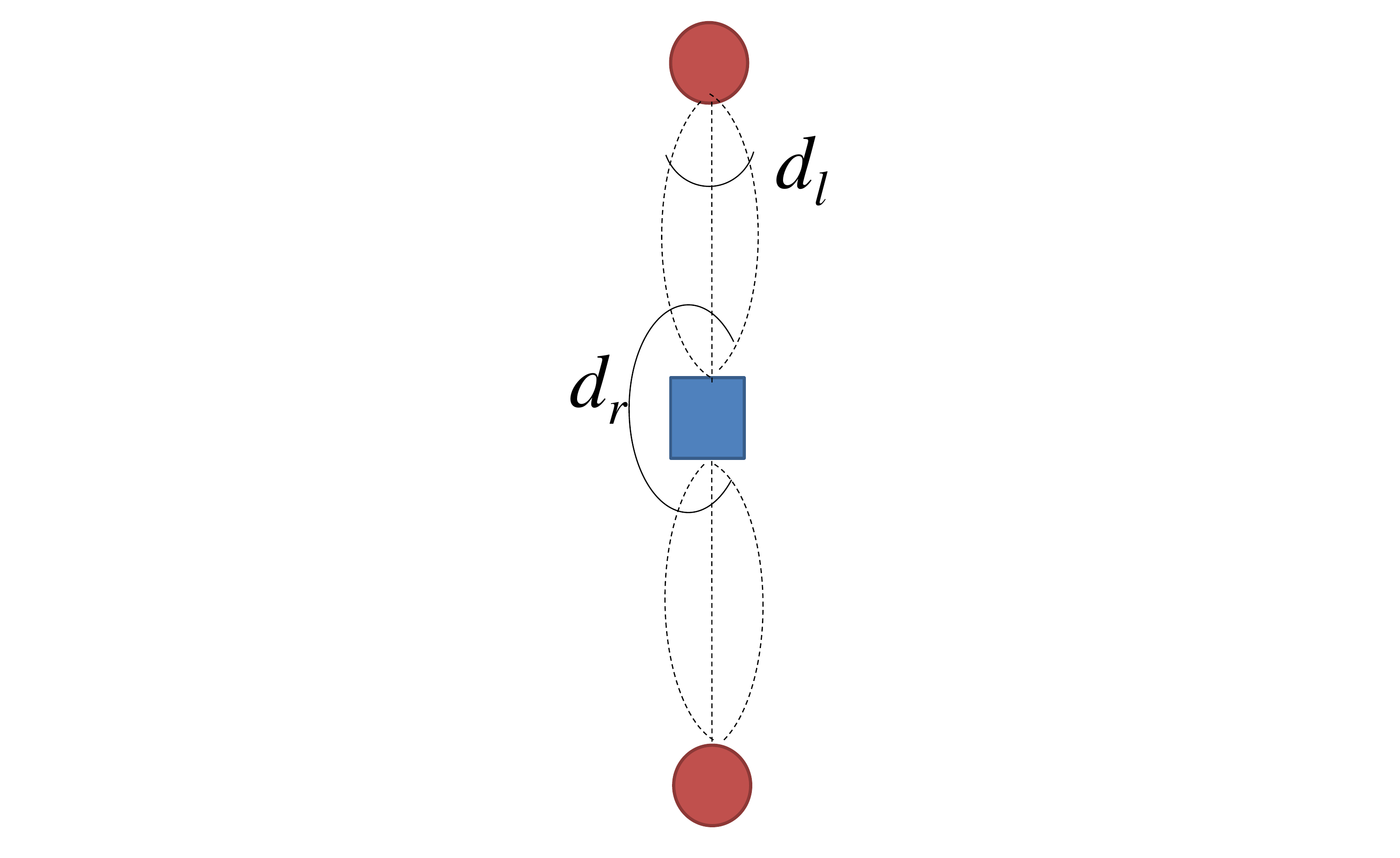}
  }
   \hspace{2em}
  \subfigure[A coupled chain of protographs for a $(3,6,L)$ SC-LDPC code.]{%
   \label{fig:sc_ldpc_ensemble}
   \includegraphics[width=0.2\textwidth]{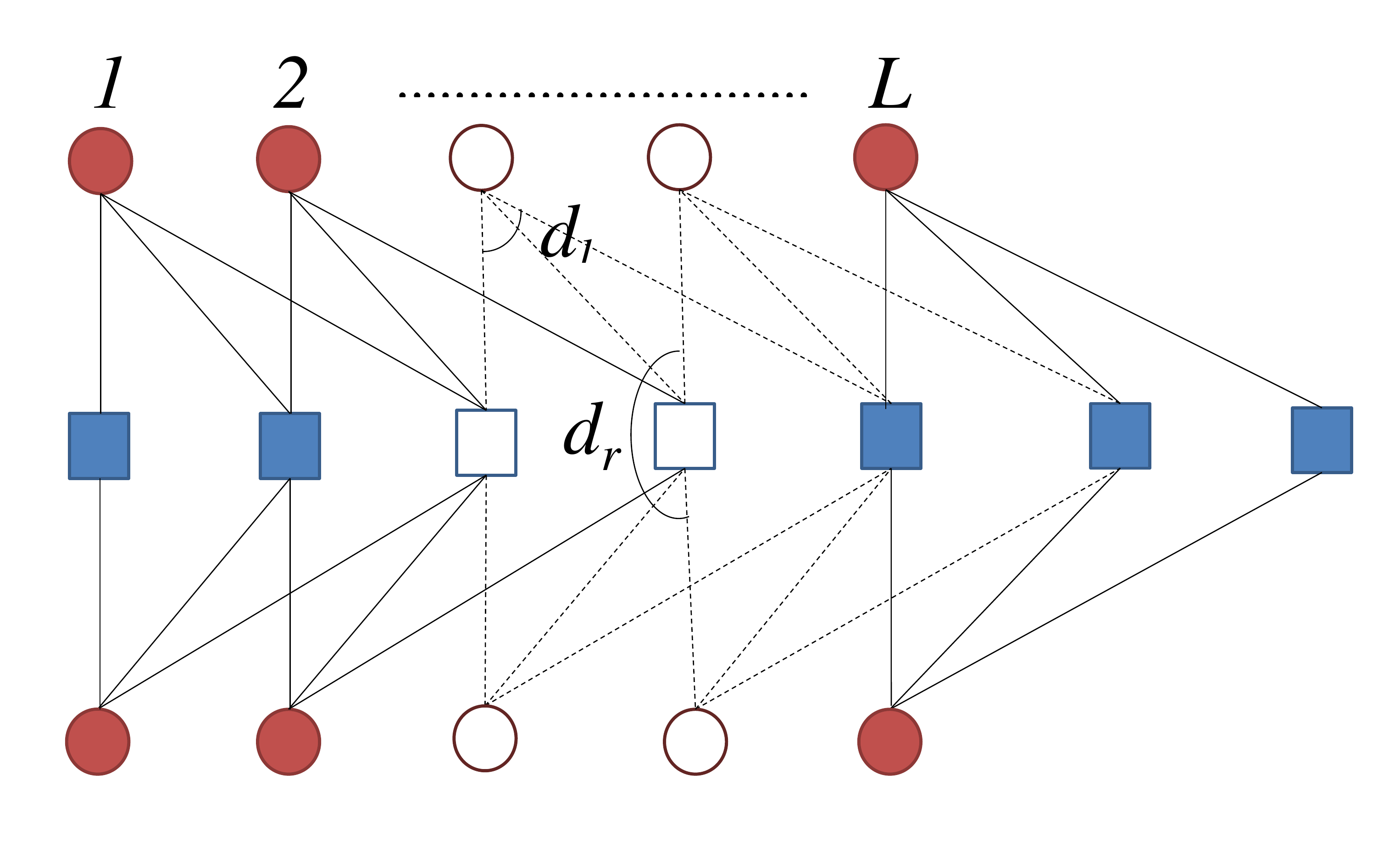}
   } \\
  \subfigure[A $(3,6,L)$ SC-LDPC Tanner graph for $M=3$.]{%
   \label{fig:sc_ldpc_graph}
   \includegraphics[width=0.5\textwidth]{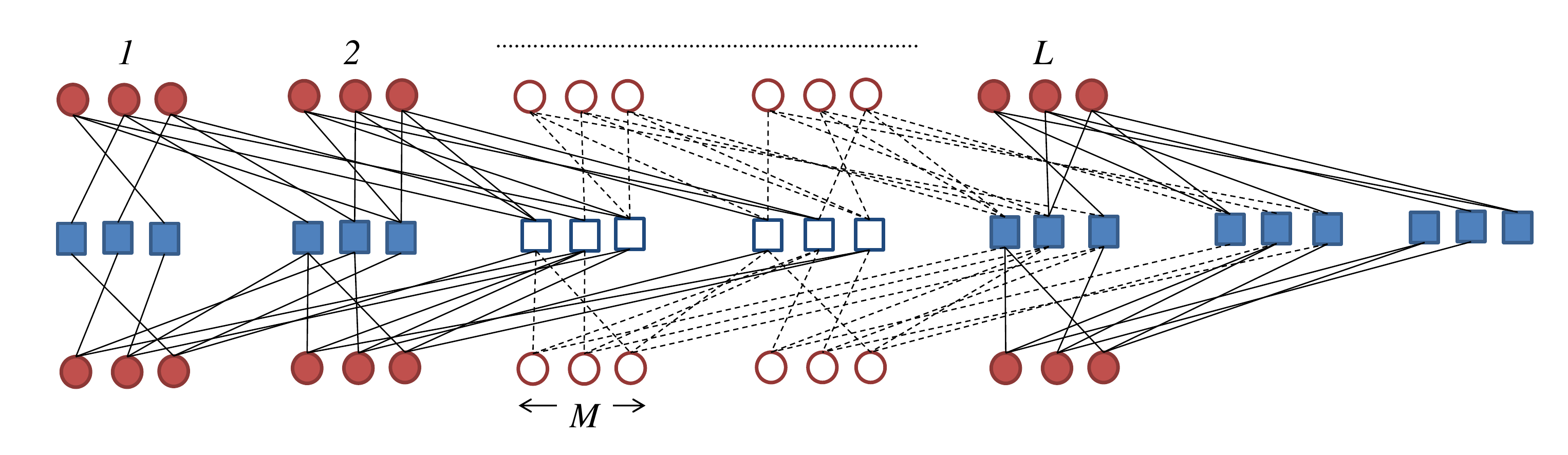}
   } \\
   \includegraphics[width=0.35\textwidth]{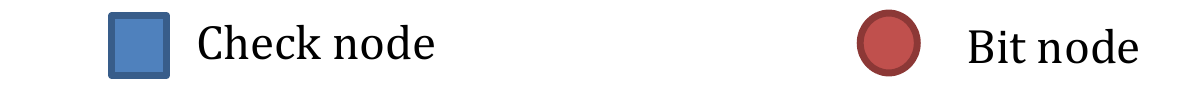}
  \caption{%
   Example of an SC-LDPC graph structure.
   }%
   \label{fig:sc_ldpc_structure}
   \end{center}
\end{figure}

Here we briefly introduce SC-LDPC codes, more details can be found in \cite{Sridharan_IT2007,Lentmaier_ITA09,Kudekar_IT11,Kudekar_ISIT11}. A sample Tanner graph structure of SC-LDPC codes (as defined in \cite{Kudekar_IT11}) is shown in Fig.~\ref{fig:sc_ldpc_structure}.  The SC-LDPC code starts with the protograph of a standard $(d_l,d_r)$-regular LDPC code, where $d_l$ and $d_r$ denotes the average bit degree and check degree respectively. Fig.~\ref{fig:sc_ldpc_protograph} shows a regular LDPC protograph for $d_l=3$ and $d_r=6$. There are $n_b = d_r/\mathrm{gcd}(d_r,d_l) = 2$ bit nodes, shown as circles and $n_c= d_l/\mathrm{gcd}(d_r,d_l) = 1$ check nodes, shown as a square, in the base protograph, where gcd stands for greatest common denominator. (When $\mathrm{gcd}(d_r,d_l)=1$ a spatially coupled code cannot be constructed using this method.)

For an SC-LDPC ensemble, a coupled chain of $L$ of these protographs (see Fig. ~\ref{fig:sc_ldpc_ensemble}) is formed by repeating the standard protograph $L$ times and connecting it once to each of the $d_l-1$ protographs to its right. There are $d_l-1$ extra check nodes added when forming the coupled chain of protographs and this reduces the rate of the resulting spatially coupled chain when compared to the original LDPC protograph code. The rate loss of course diminishes as $L$ is increased.

A particular code of length $N$ is formed from the ($d_l,d_r,L$) protograph chain by creating $M$ copies of every node and edge in the coupled chain \footnote{We use $M$ here as defined in \cite{Sridharan_IT2007,Lentmaier_ITA09} and others. Note that in the terminology of \cite{Kudekar_IT11}, $M$ is alternatively defined as the number of copies of each node divided by $n_b$.}. If a particular bit node was coupled to a particular check node in the original chain each of its copies is then connected to one of the $M$ copies of that check node in the final code. The choice of which copy each node is connected to will depend on the constructions of the code. Like any standard LDPC code, SC-LDPC codes can be constructed randomly or incorporating structure of some form. In the following section we present a QC algorithm for the construction of SC-LDPC codes. In a QC SC-LDPC code, all of the $M$ copies of a particular edge in the protograph are specified using a circulant. The length $N$ and rate $r$ of the SC-LDPC code in Fig.~\ref{fig:sc_ldpc_graph} are computed as follows:
\begin{equation} \label{eq:code_length}
N = n_bML
\end{equation}
\begin{equation} \label{eq:code_rate}
r = \frac{n_b L-n_c(L+d_l-1)}{n_b L}
\end{equation}
where, $n_b$ and $n_c$ are the number of bit and check nodes in the original protograph respectively.

\section{Memory-Efficient Quasi-Cyclic SC-LDPC codes}
The quasi-cyclic (QC) construction of LDPC codes has various advantages including simple encoding \cite{Mahdi2011}, parallel decoding and memory efficiency in storing the matrix elements in the decoder \cite{Yongmei2008}. Hence, for the same advantages, it is desirable to have SC-LDPC codes also in QC form. A dual-core programmable QC LDPC convolutional decoder has been designed in \cite{Tavares08}. However, the design uses a time-invariant code whose performance is poor compared to time-variant convolutional LDPC codes \cite{Pusane2011}. The construction of non-periodic time-varying QC SC-LDPC code is presented in \cite{Hagiwara2011,Mitchell2010}. \cite{Hagiwara2011} discusses the performance improvements achieved by spatial coupling; particularly for quantum LDPC codes. Whereas \cite{Mitchell2010} presents techniques to improve the upper bound on the minimum Hamming distance of members of the QC sub-ensembles. Recently, a time-varying periodic covolutional LDPC decoder has been designed from a QC LDPC block code that achieves a throughput of 2 $Gb/s$ \cite{Chiu2013}.

In contrast to the above variations of convolutional LDPC codes, this paper proposes an innovative technique to construct time-varying QC SC-LDPC code with periodicity, $T \geq~ 1$, and also without the need of a QC LDPC block code. It is shown that by using periodic codes ($T\ll~ L$), the BER performance is as good as that of non-periodic time-varying codes. As an added advantage to the reduced complexity, we also show that there is a significant reduction in the hardware requirements of the decoder, especially the Block RAMs (BRAM) in an FPGA.
%It is well known that a quasi-cyclic structure in the LDPC codes inherently offers a huge reduction in the complexity for designing an encoder or a decoder \cite{Yun2010}.
%The construction of QC SC-LDPC codes and minimum distance analysis of such codes are discussed in \cite{Mitchell2010}.

While designing an LDPC decoder, it is essential to store the structure of the LDPC matrix in the hardware to carry out the decoding process. The structure is normally stored in the form of BRAMs because of the enormous data required, particularly for large LDPC codes \cite{Xiaoheng2010}. Compared to the codes that are constructed using PEG algorithms \cite{Zhengang05}, the QC based technique requires significantly less memory to store the matrix structure. This is because the latter requires memory for storing only the circulant, instead of the entire matrix. However, a special circulant-processing block is needed for cyclic-shifting of the identity matrix for a given circulant to realize the actual and complete LDPC matrix in the decoder hardware \cite{Chandrasetty2011b}. Therefore, decoders using QC based LDPC matrices require less memory but use additional logic elements for special operations.

The QC SC-LDPC matrix consists of a chain of circulants making up the $L$ protographs. For a $(d_l,d_r = 2d_l)$-regular code, (i.e., a rate $\frac{1}{2}$ code) each protograph in the chain corresponds to a set of $2d_l$ individual circulants.

The staircase structure of the SC-LDPC codes offers a unique opportunity to optimize the QC based technique to further reduce the hardware resources required for the decoder. We investigate whether a fixed set of circulants in the matrix (i.e., fewer than $L$ independent circulant sets) can be reused without affecting the decoding performance. To this end, we consider SC-LDPC codes constructed with one, two and three repeating sets of circulants. In other words, we can say time-varying QC SC-LDPC codes with periodicity, $T$=1, 2 and 3. Fig.~\ref{fig:qc_sc_1}, Fig.~\ref{fig:qc_sc_2} and Fig.~\ref{fig:qc_sc_3} show one, two and three repeating sets respectively for a rate $\frac{1}{2}$ $(3,6,L)$ QC SC-LDPC code. Each unique circulant is represented by a unique letter of the alphabet, and similarly repeated circulants are represented by the same letter. Note that in the given example, the QC SC-LDPC codes with $T$=1, 2 and 3 require a limited number of 6, 12 and 18 circulants only. Thus the QC SC-LDPC code is completely specified by a very small number of shift values (6, 12 or 18), significantly reducing the memory requirements in storing the SC-LDPC matrix in the decoder.

\begin{figure}
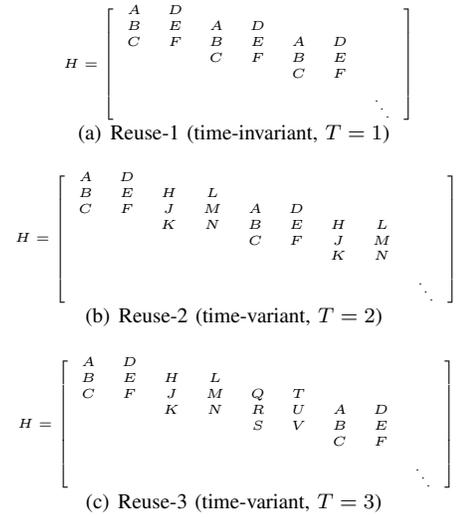

 \begin{center}
  \subfigure[Reuse-1 (time-invariant, $T=1$)]{%
   \label{fig:qc_sc_1}
   \tiny{ $ H = \left[ \begin{array}{ccccccc} A  & D &    &    &  & &\\
                                        B  & E & A  & D  &  & &\\
                                        C  & F & B  & E  &A &D &\\
                                           &   & C  & F  &B &E &  \\
                                           &   &    &    &C &F &\\
                                           &   &    &    &  &  & \ddots \\
\end{array}
\right]
$}
  }
  \hfill
  \subfigure[Reuse-2 (time-variant, $T=2$)]{%
   \label{fig:qc_sc_2}
     \tiny{ $ H = \left[ \begin{array}{cccccccccc} A  & D &    &   &   &   &   &   &   \\
                                                   B  & E & H  & L &   &   &   &   &\\
                                                   C  & F & J  & M & A & D &   &   &\\
                                                      &   & K  & N & B & E & H & L &\\
                                                      &   &    &   & C & F & J & M & \\
                                                      &   &    &   &   &   & K & N & \\
                                                      &   &    &   &   &   &   &   & \ddots\\
\end{array}
\right] $ }
   }%
     \end{center}
      \begin{center}
   \subfigure[Reuse-3 (time-variant, $T=3$)]{%
   \label{fig:qc_sc_3}
  \tiny{ $ H = \left[ \begin{array}{cccccccccc}    A  & D &    &   &   &   &   &  &    \\
                                                  B  & E & H  & L &   &   &   &  &  \\
                                                  C  & F & J  & M & Q & T &   &  & \\
                                                     &   & K  & N & R & U & A & D & \\
                                                     &   &    &   & S & V & B & E    &   \\
                                                     &   &    &      &   &   & C & F  & \\
                                                     & & & & & & & &  \ddots\\
\end{array}
\right] $ }
   }%
 \end{center}
  \caption{%
   SC-LDPC parity-check sub-matrices from a $(3,6)$-regular protograph with quasi-cyclic column reuse. Each unique letter represents a unique circulant. Empty elements represent the all zero matrix.
   }%
   \label{fig:opt_qc_sc_ldpc}
 \end{figure}

\subsection{The effect of circulant reuse on code girth}

Tanner \cite{Tanner2001} and Fossorier \cite{Fossorier2004} have shown that QC LDPC codes defined by an array of circulant matrices have a maximum possible girth of 12. Protograph LDPC codes, such as SC-LDPC codes, allow zero matrices in place of some circulant matrices and so are not limited to this bound. However, for these more general LDPC codes Kim et al. in \cite{Kim_IT2007} have shown that the girth of the QC code can be bounded by the girth of its base protograph. Kim et al. defined \emph{inevitable cycles} in a QC-LDPC code as those cycles that always exist in the code regardless of the choice of circulant permutations. They then found all the subgraph patterns of protographs which lead to inevitable cycles of size up to girth 20.

It is easy to see that the submatrix
\[
P_{12} =  \left[ \begin{array}{cc}
     1 & 1 \\
     1 & 1 \\
     1 & 1 \\
   \end{array} \right]
\]
exists in the base matrix of any SC-LDPC code with $d_l>=4$ and any code with $d_l=3$ that has $n_b >= 2$. Thus, applying the results of \cite{Kim_IT2007}, these SC-LDPC codes have inevitable cycles of length 12 and hence an upper bound of 12 on their girth.

Kim et al. also showed that a protograph with girth $g \geq 4$ cannot contain inevitable cycles of length smaller than $3g$. Thus if the circulants are large enough and chosen appropriately, a QC SC-LDPC code with girth $12$ can be found. However, choosing the shift value of the circulants in QC-LDPC codes to avoid all non-inevitable cycles, and thus obtain a girth equal to the minimum inevitable cycle length, is not trivial and few algebraic constructions for QC girth 12 codes have been found. One notable exception is an algebraic construction for a class of (3,5)-regular QC-LDPC codes of Tanner \cite{Tanner2001} which obtain girth 12 for $p \times p$ circulant matrices for certain choices of prime $p$ \cite{Kim_IT2006}.

For a QC SC-LDPC code with the reuse of $T$ columns of circulants as in Fig.~\ref{fig:opt_qc_sc_ldpc}, the choice of circulants is restricted by the reuse factor and so we define reuse-$T$ inevitable cycles as those cycles that always exist in the code when circulant permutations are reused, regardless of the choice of those circulant permutations.

Adapting the notation of \cite{Fossorier2004}, we define the parity-check matrix $H$ for a general SC-LDPC code as shown in Fig.~\ref{fig:fossorier_sc_ldpc},
\begin{figure*}[ht!]
\begin{center}
\[ \label{LDPC_CC} H_{cc} = \left[
\begin{array}{ccccccccccc}
I_{(p_{1,1})}   & \hdots & I_{(p_{1,n_b})}   &                       &        &                      &                       &       &&        \\
I_{(p_{2,1})}   & \hdots & I_{(p_{2,n_b})}   & I_{(p_{2,n_b+1})}     & \hdots & I_{(p_{2,2n_b})}     &                       &       &&        \\
\vdots          & \hdots & \vdots            & I_{(p_{3,n_b+1})}     & \hdots & I_{(p_{3,2n_b})}     & I_{(p_{3,2n_b+1})}    &\hdots &I_{(p_{3,3n_b})}&        \\
I_{(p_{d_l,1})} & \hdots & I_{(p_{d_l,n_b})} & \vdots                & \hdots & \vdots               & I_{(p_{4,2n_b+1})}    &\hdots &I_{(p_{4,3n_b})}&        \\
                &        &                   & I_{(p_{d_l+1,n_b+1})} & \hdots & I_{(p_{d_l+1,2n_b})} & \vdots                &\hdots &\vdots& \ddots \\
                &        &                   &                       &        &                      & I_{(p_{d_l+2,2n_b+1})}&\hdots &I_{(p_{d_l+2,3n_b})}&        \\
\end{array}
\right]
\]
\end{center}
 \caption{General representation of SC-LDPC code.}
 \label{fig:fossorier_sc_ldpc}
\end{figure*}
%\[ \label{LDPC_CC} H_{cc} =\left[
%\fontsize{2}{2}
%\setlength{\arraycolsep}{0.1pt}
%  \begin{array}{ccccccccc}
%     I_{(p_{1,1})}   & \hdots & I_{(p_{1,n_b})}   &                       &        &                      &       \\
%     I_{(p_{2,1})}   & \hdots & I_{(p_{2,n_b})}   & I_{(p_{2,n_b+1})}     & \hdots & I_{(p_{2,2n_b})}     &       \\
%     \vdots          & \hdots & \vdots            & I_{(p_{3,n_b+1})}     & \hdots & I_{(p_{3,2n_b})}     &       \\
%     I_{(p_{d_l,1})} & \hdots & I_{(p_{d_l,n_b})} &  \vdots               & \hdots & \vdots               &       \\
%                     &        &                   & I_{(p_{d_l+1,n_b+1})} & \hdots & I_{(p_{d_l+1,2n_b})} & \ddots\\
%                     &        &                   &                       &        &                      &       \\
% \end{array}
%\right]
%\]
where $I_{(p_{x,y})}$ is a circulant matrix with a shift value $p_{x,y}$. I.e. $I_{(p_{x,y})}$ represents the circulant permutation matrix with a one at column $(r + p_{x,y}) \text{ mod } M$ for row $r$, $0 \leq r \leq M - 1$, and zero elsewhere. For notational clarity, we assume that the SC-LDPC code has $n_c = 1$, however the girth results hold for all $n_c$.

A cycle of length $2i$ in $H$ is described by a sequence of $2i$ positions $H_{x,y}$ such that: 1) each consecutive position is obtained by changing, alternatively, the row or column index of the previous position; and 2) all positions are distinct, except the first and last ones. Thus two consecutive positions in any cycle belong to distinct circulant permutation matrices which are either
in the same row, or in the same column. A length $2i$ cycle exists if and only if \cite{Fossorier2004}:
\begin{equation} \label{Fossorier_def} \sum_{k=0}^{i-1} \Delta_{x_k,x_{k+1}} (l_k) = 0 \mathrm{ mod } M
\end{equation}
where $x_0 = x_i$, $x_k \neq x_{k+1}$,  $l_k \neq l_{k+1}$ and
\[ \Delta_{x_k,x_j} (l) =  p_{x_k,l} - p_{x_j,l}.
\]

Given this notation we can now define the inevitable cycles in the QC SC-LDPC codes with circulant reuse.
\begin{lemma} \label{lemma:1}
A QC SC-LDPC code with reuse $T = 1$, and ${d_l \geq 3}$ has a reuse-1 inevitable cycle of length $6$ and thus a maximum girth of $6$.
\end{lemma}
\begin{proof}
Following \eqref{Fossorier_def}, a cycle is described by the positions \begin{eqnarray*}  H_{x,y} &=& [p_{2,1},p_{2,n_b+1},p_{4,n_b+1},p_{4,2n_b+1}, \\ & & p_{3,2n_b+1},p_{3,1}] \end{eqnarray*} since, by the reuse-1 construction, $p_{2,1}=p_{4,2n_b+1}$, $p_{2,n_b+1}=p_{3,2n_b+1}$, and $p_{4,n_b+1}=p_{3,1}$.
\end{proof}

\begin{lemma} \label{lemma:2}
A QC SC-LDPC code with reuse $T = 2$, and ${d_l \geq 4}$ has a reuse-2 inevitable cycle of length $8$ and thus a maximum girth of $8$.
\end{lemma}
\begin{proof}
Following \eqref{Fossorier_def} a cycle is described by the positions \begin{eqnarray*}  H_{x,y} &=& [p_{2,1},p_{2,n_b+1},p_{4,n_b+1},p_{4,2n_b+1}, \\ & & p_{6,2n_b+1},p_{6,3n_b+1},p_{4,3n_b+1},p_{4,1}] \end{eqnarray*} since, by the reuse-2 construction, $p_{2,1}=p_{4,2n_b+1}$, $p_{2,n_b+1}=p_{4,3n_b+1}$, $p_{4,n_b+1}=p_{6,3n_b+1}$ and  $p_{6,2n_b+1}=p_{4,1}$.
\end{proof}

\begin{lemma} \label{lemma:3}
A QC SC-LDPC code with reuse $T = 3$, and ${d_l \geq 4}$ has a reuse-3 inevitable cycle of length $10$ and thus a maximum girth of $10$.
\end{lemma}
\begin{proof}
Following \eqref{Fossorier_def} a cycle is described by the positions \begin{eqnarray*} H_{x,y} &=& [p_{2,1},p_{2,n_b+1},p_{5,n_b+1},p_{5,3n_b+1},p_{6,3n_b+1},  \\ & & p_{6,4n_b+1}, p_{5,4n_b+1},p_{5,n_b+1},p_{3,n_b+1},p_{3,1}] \end{eqnarray*} since, by the reuse-3 construction, $p_{2,1}=p_{5,3n_b+1}$, $p_{2,n_b+1}=p_{5,4n_b+1}$, $p_{3,1}=p_{6,3n_b+1}$ and $p_{3,n_b+1}=p_{6,3n_b+1}$.
\end{proof}

Although the proof only requires that we find one 6-cycle (respectively 8-cycle or 10-cycle) which exists for any choices of circulants, in fact every column of $H$ (and hence every codeword bit) is involved in 6-cycles (respectively 8-cycles or 10-cycles) in QC SC-LDPC codes with reuse-1 (respectively reuse-2 or 3) regardless of which circulants are chosen.

Consequently, the reuse-1 codes not only have a poor girth but also have a very large number of cycles of the minimum length. Similarly for the reuse-2 codes, while a girth of 8 is not necessarily problematic for LDPC codes if there are only a few such cycles, the very large number of 8-cycles in the QC SC-LDPC is certainly detrimental for the performance of the sum-product decoder when longer codes are considered.

%When the quasi-cyclic code has $d_l=3$, 8-cycles can be avoided, however.
%The proof for the cases $T=1$ and $3$ follows a similar argument.

%\begin{lemma}
%Given a quasi-cyclic SC-LDPC code with reuse-3, $d_l\geq 4$ and $L\geq 2$, the girth of this code is at most 10. A quasi-cyclic SC-LDPC code with reuse-3, $d_l=3$ and $L\geq 2$, has a girth of at most 12.
%\end{lemma}

\section{Performance of SC-LDPC codes}

The QC SC-LDPC codes were simulated to evaluate the BER performance on a binary input additive white Gaussian noise (BI-AWGN) channel. We used multi-edge density evolution to compute the threshold of SC-LDPC codes, shown in Fig.~\ref{fig:ldpc_awgn_threshold}, and noted that $L \geq 33$ is necessary to achieve thresholds better than that of $(3,d_r)$ standard LDPC codes and that codes with $d_l=4$ have an improved threshold over  codes with $d_l=3$ as $L$ is increased. Also, as SC-LDPC codes are known to have good performance for very long codes, we compared decoding performance with code length 25K, 100K and 250K as shown in Fig.~\ref{fig:ber_sc_ldpc_var_mat}. Given these results, a fairly long SC-LDPC code of 103,200-bits ($\approx 100K$) with $d_l=4$, $L = 129$, $M= 400$ and $r=0.488$ is considered for our following simulation results. However, we note that even longer codes would perform noticeably better. Simulations were carried out using software models on a BI-AWGN channel. The sum-product algorithm was used for decoding with a maximum of 1000 iterations, and the simulation was run until at least 50 word errors were accumulated.

\begin{figure}
 \begin{center}
  \includegraphics[width=0.45\textwidth]{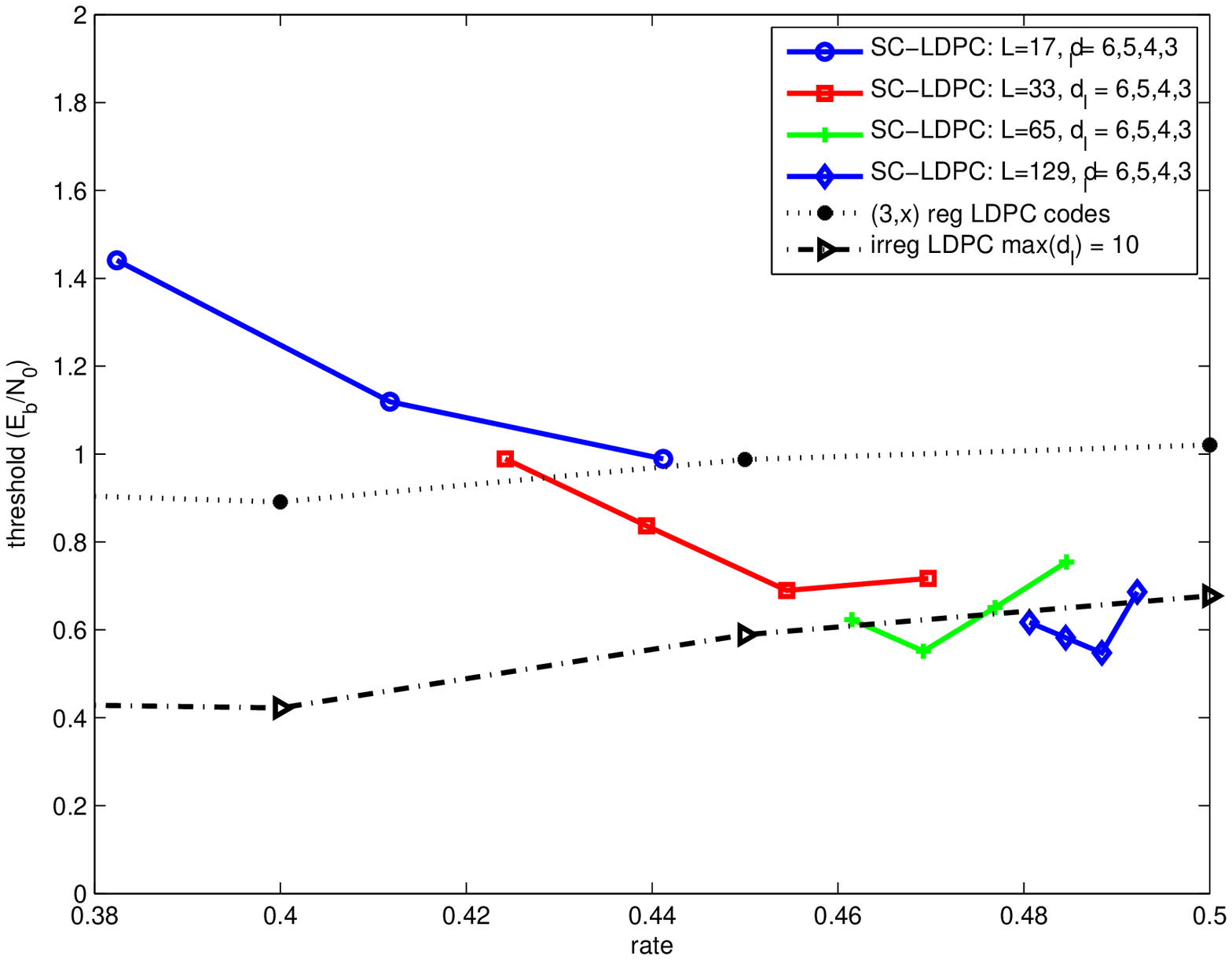}
 \end{center}
 \caption{Thresholds for LDPC codes from density evolution.}
 \label{fig:ldpc_awgn_threshold}
\end{figure}

\begin{figure}
 \begin{center}
   \includegraphics[width=0.5\textwidth]{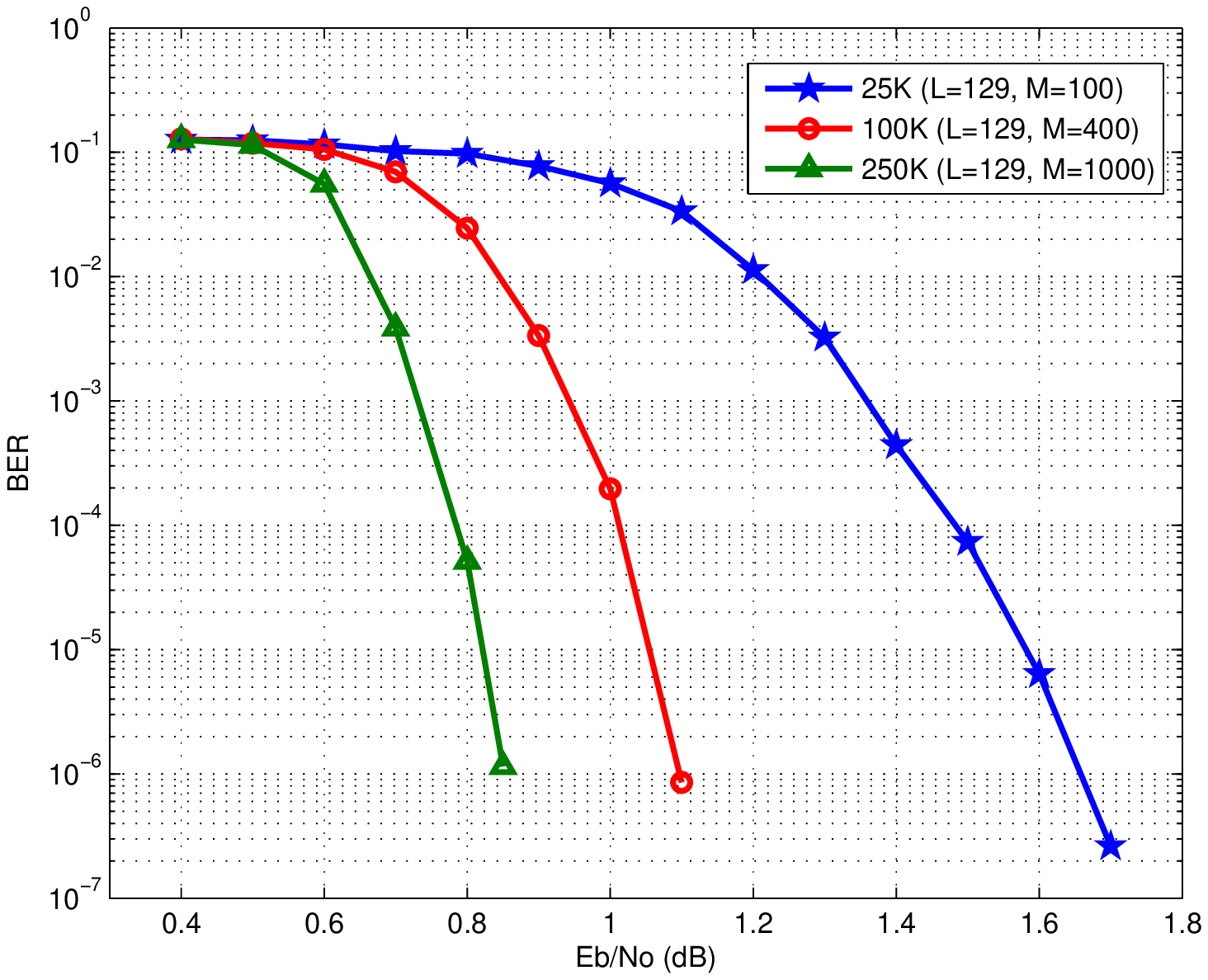}
  \end{center}
  \caption{BER Performance of SC-LDPC codes for various code lengths. Code rate is $0.488$.}
  \label{fig:ber_sc_ldpc_var_mat}
 \end{figure}

 \begin{figure}
 \vspace{-3mm}
 \begin{center}
  \includegraphics[width=0.5\textwidth]{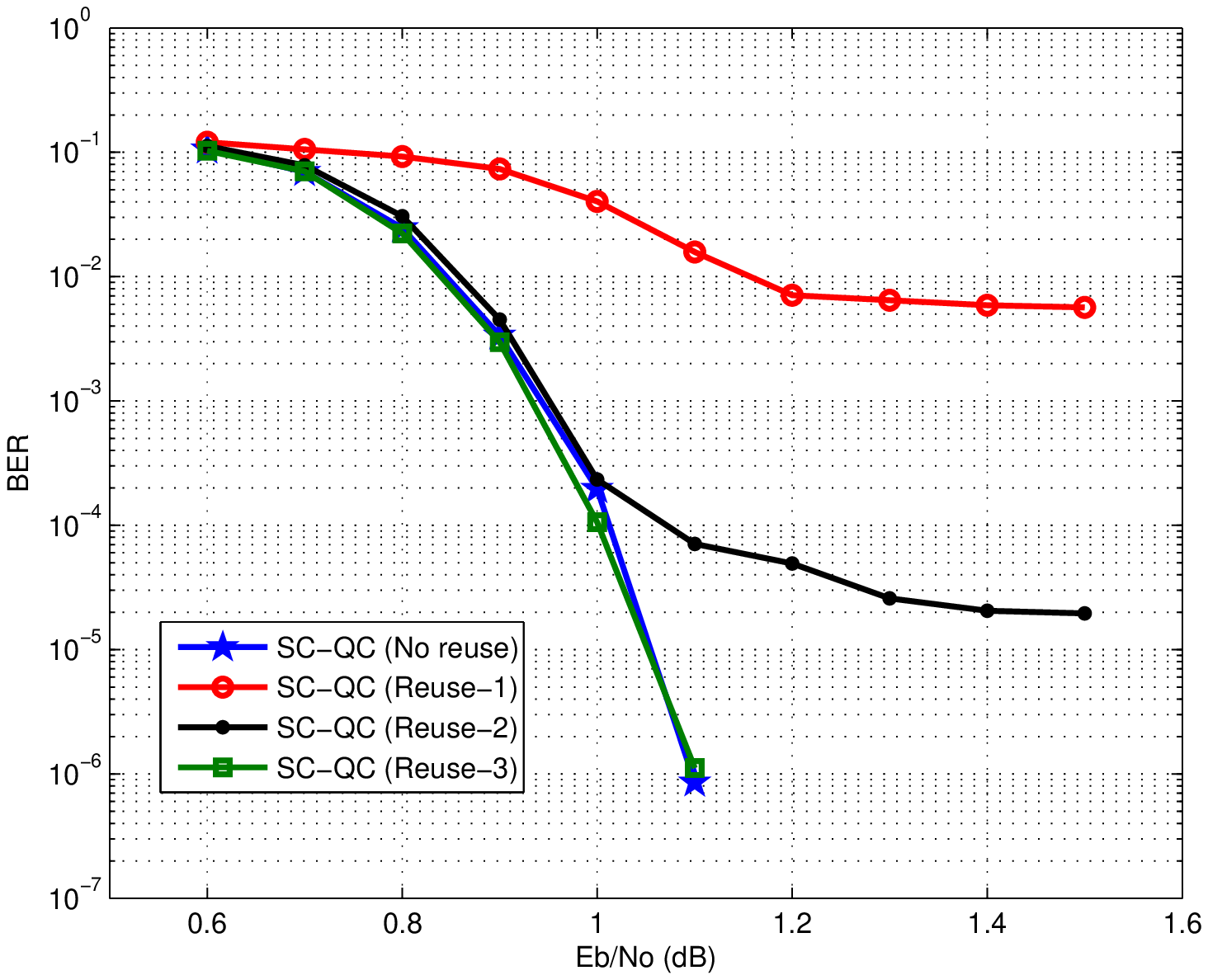}
 \end{center}
 \caption{BER performance of quasi-cyclic regular $(4,8,129)$ SC-LDPC codes of length 100K, rate $0.488$, $M=400$.}
 \label{fig:ber_sc_qc}
 \end{figure}

\begin{figure}
 \begin{center}
   \includegraphics[width=0.5\textwidth]{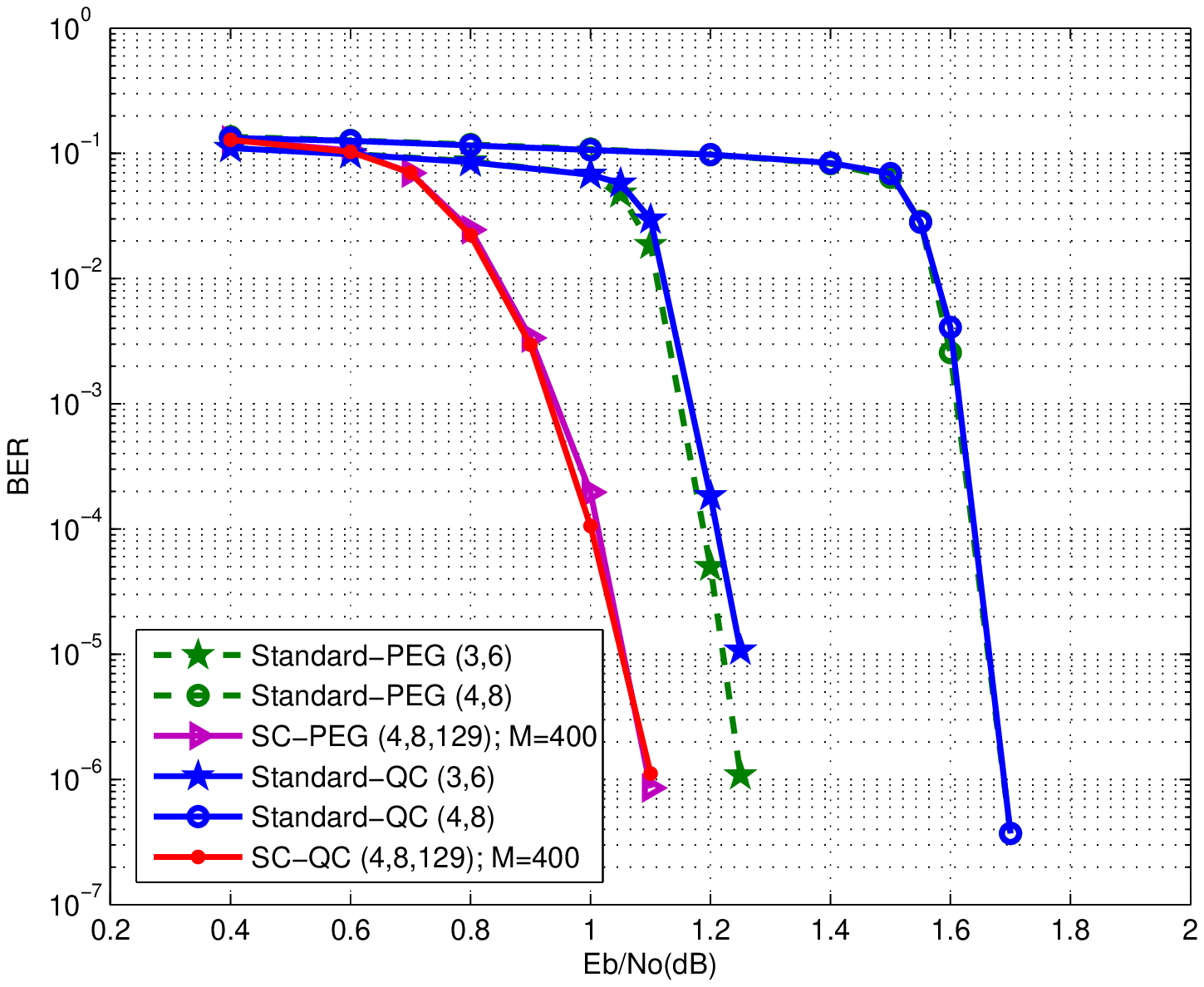}
  \end{center}
  \caption{BER performance of standard and SC-LDPC codes of length 100K. Code rates are $0.5$ and $0.488$ for standard-LDPC (PEG and QC) and SC-LDPC codes respectively.}
  \label{fig:ber_ldpc_matrices}
 \end{figure}

Fig.~\ref{fig:ber_sc_qc} shows a comparison of QC SC-LDPC codes with different levels of circulant reuse. In the construction process for all codes the circulants are chosen to avoid girth 4 in the resulting matrix. As would be expected due to girth limitations, the quasi-cyclic matrix with reuse of one circulant column (i.e., time invariant) and two circulant columns (periodic time-varying with period 2) have poor performances. However, by reusing three circulant columns (periodic time-varying with period 3), the BER performance we obtained is as good as standard (non-periodic time-varying) QC SC-LDPC codes down to a bit error rate of $10^{-6}$.

Lastly, in Fig.~\ref{fig:ber_ldpc_matrices} we compare the BER performance of the QC SC-LDPC codes with non-quasi-cyclic SC-LDPC codes and with standard LDPC block codes. The non-quasi-cyclic SC-LDPC codes are constructed using a PEG algorithm modified for spatial coupling. The standard LDPC codes are constructed according to PEG \cite{Xiao-Yu2001} and QC \cite{Fossorier2004} techniques. From Fig.~\ref{fig:ber_ldpc_matrices}, it is clear that the waterfall regions of SC-LDPC codes are better than standard LDPC codes (as predicted by density evolution, see Fig.~\ref{fig:ldpc_awgn_threshold}). Also, the performance of the QC based SC-LDPC codes is very similar to that of PEG based codes.

\section{Estimation of hardware requirements for using SC-LDPC codes}

\begin{figure}
 \begin{center}
   \includegraphics[width=0.5\textwidth]{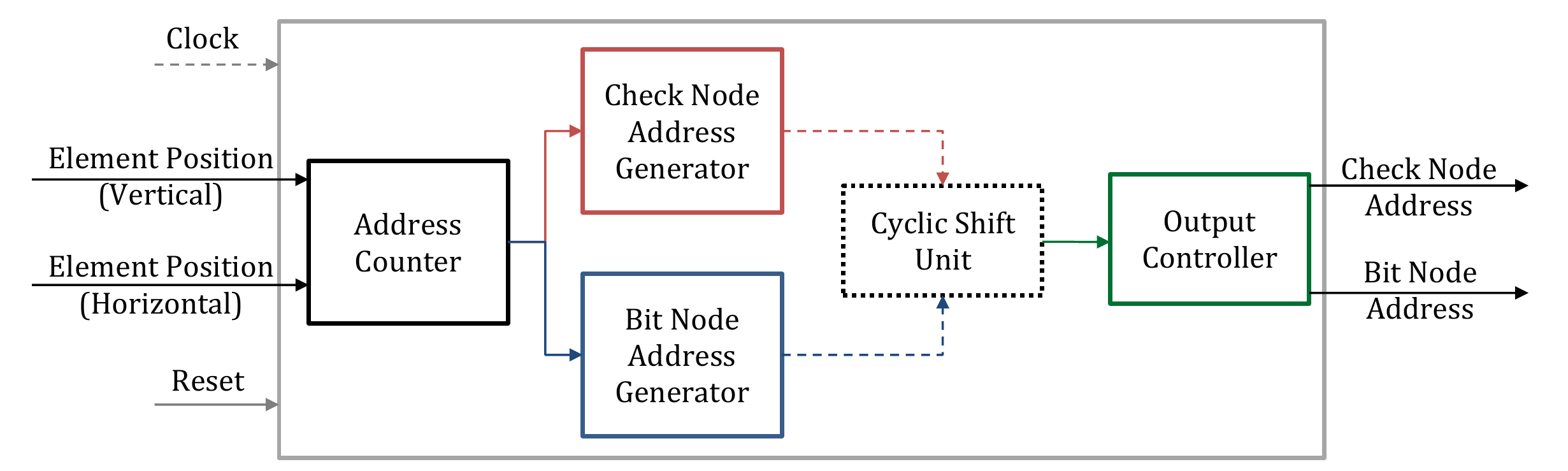}
   \end{center}
   \caption{Block diagram of the hardware model for storing and processing QC SC-LDPC codes.}
   \label{fig:sc_ldpc_fpga}
\end{figure}

The new QC SC-LDPC codes are compared with random SC-LDPC codes by designing a decoder hardware model. The model consists of units that are essential for storing the contents of the LDPC matrix and generating the appropriate address locations in a decoder. The block diagram of the hardware model is shown in Fig.~\ref{fig:sc_ldpc_fpga}. The hardware model consists of a clock and a  synchronous reset as inputs. It also consist of an address counter and output controller for sequencing the input and output data respectively. The bit and check node address generators are responsible for storing the LDPC matrix information and generating appropriate addresses for the decoder. Note that the QC SC-LDPC codes require an additional \emph{Cyclic Shift} unit to generate appropriate addresses based on the circulants for decoding, as shown in Fig.~\ref{fig:sc_ldpc_fpga}. Whereas PEG based SC-LDPC codes do not require any such unit, since the complete set of matrix elements are stored in hardware memory.

The hardware models for both PEG and QC SC-LDPC codes have been designed and synthesized using Verilog HDL for FPGA implementation. The designs are placed and routed for Xilinx Kintex-7 (XC7K355T) FPGA  with LDPC codes of length 100K and 25K. The estimates of FPGA hardware requirements and maximum clock frequency achievable for the designs are shown in Table~\ref{tab:fpga_requirements}. As expected, a large number of BRAMs are utilized by the PEG based codes compared to QC SC-LDPC codes, to store the matrix elements. With slightly increased logic units - registers and look-up tables (LUT), the standard QC codes offer a significant saving (up to 43 times in the case of 100K code lengths) of BRAMs compared to PEG SC-LDPC codes. As stated earlier, the increased logic requirements are due to the \emph{Cyclic Shift} unit in the QC based SC-LDPC hardware models. Further, it is also noted that QC SC-LDPC codes with reuse-3 do not require BRAMs. This is due to few set of reusable circulants in the code, that can be easily stored in the LUTs. Elimination of BRAMs (which normally have large access delays) also results in a significant improvement in the speed of operation (up-to 40\% increase in the maximum operating clock frequency) compared to PEG or standard QC SC-LDPC codes (for 100K codes).

\begin{table}
\vspace{3em}
\caption{FPGA resources required for various SC-LDPC codes}
\label{tab:fpga_requirements}
\centering
\resizebox{0.45\textwidth}{!}{
\begin{tabular}{l||cc|cc|c}
\hline
\hline
LDPC codes & \multicolumn{2}{c}{Std-PEG} &  \multicolumn{2}{c}{Std-QC} & QC (Reuse-3)\\[0.4ex]
\hline
Code length & 100K & 25K & 100K & 25K & 100K / 25K \\[0.5ex]
\hline
Registers   & 286  & 268 & 320  & 302 & 322 \\
LUTs        & 185  & 170 & 890  & 819 & 609 \\
Slices      & 143  & 139 & 357  & 356 & 214 \\
RAM         & 528  & 116 & 12   & 12  &  -  \\
Clock (MHz) & 188  & 215 & 180  & 182 & 264 \\
\hline
\end{tabular}}
\end{table}

\section{Conclusion}
This paper has presented the construction of periodic time-varying SC-LDPC codes using the quasi-cyclic technique. It also demonstrates the threshold advantages achievable by SC-LDPC codes over standard LDPC codes. Memory optimized QC SC-LDPC codes are introduced to significantly reduce the complexity of the decoder for storing the matrix elements. It is shown that by reusing the circulant columns in the SC-LDPC matrix (with reuse-3), it is possible to obtain a memory efficient decoder without noticeably affecting the decoding performance. The advantages of using the optimized QC SC-LDPC codes have been demonstrated by designing a hardware model, which shows substantial reduction in memory requirements and a significant improvement in the operating clock frequency of the decoder.

%\small
%\bibliographystyle{ieeetran}
%\bibliographystyle{wileyj}
%\bibliography{Abrv,QC_SC_LDPC}

\end{document}